\documentclass[runningheads]{llncs}
\usepackage[T1]{fontenc}
\usepackage{graphicx}
\usepackage[page]{appendix}
\usepackage[hyphens]{url}
\usepackage{tikz}
\usetikzlibrary{decorations.pathreplacing}

\usepackage{algorithm}
\usepackage{algpseudocode}
\usepackage{float}
\usepackage{mathtools}
\usepackage{bm}

\DeclarePairedDelimiter{\ceil}{\lceil}{\rceil}

\usepackage[normalem]{ulem}
\usepackage{placeins}

\begin{document}

\title{A Linear-Time 1.5-Approximation for Broadcasting in k-Cycle Graphs}

\author{Jeffrey Bringolf \inst{1}\orcidID{0009-0003-0928-9146} \and
Anne-Laure Ehresmann\inst{2}\orcidID{0000-0002-6116-3212} \and
Hovhannes A. Harutyunyan\inst{1}\orcidID{0000-0001-7260-4186}}
\authorrunning{J. Bringolf, A.-L. Ehresmann, and H. A. Harutyunyan}

\institute{Department of Computer Science and Software Engineering, Concordia \email University, Montreal, Canada \\ \email{bringolfj@gmail.com}\\ \email{haruty@cs.concordia.ca} \and
YesWeHack, Singapore, Singapore \\ \email{al@ehresmann.eu}}

\maketitle     

\begin{abstract}
Broadcasting is an information dissemination primitive where a message originates at a node (called the originator) and is passed to all other nodes in the network. Broadcasting research is motivated by efficient network design and determining the broadcast times of standard network topologies. Verifying the broadcast time of a node $v$ in an arbitrary network $G$ is known to be NP-hard. Additionally, recent findings show that the broadcast time problem is NP-hard in several highly restricted subfamilies of cactus graphs. The most restrictive of these families is known as \emph{$k$-cycle graphs} or \emph{flower graphs} and is the focus of this paper. We present a simple $(1.5-\epsilon)$-approximation algorithm for determining the broadcast time of networks modeled using $k$-cycle graphs, where $\epsilon > 0$ depends on the structure of the graph.
\end{abstract}
\section{Introduction}
\label{sec:introduction}
Broadcasting is a crucial process for information dissemination in interconnected networks. As a result, a variety of network models have been studied at length \cite{problems-of-communication,limited-memory-broadcasting,gossiping-chapter,broadcasting-unicyclic-graphs,gossiping,interconnection-networks,general-graph-approximation}. These models have varying rules describing the means by which information can be disseminated. They may have different numbers of originators, numbers of receivers at each time unit, distances of each call, numbers of destinations, and other characteristics of the network such as knowledge of the neighbourhood available to each node. In this paper, we will focus on the classical model of broadcasting, sometimes referred to as \emph{telephone broadcasting}. The network is modeled as an undirected connected graph $G = (V, E)$, where $V(G)$ and $E(G)$ denote the vertex set and the edge set of $G$, respectively. The classical model is described by the following assumptions.
\begin{enumerate}
    \item The broadcasting process is split into discrete time units.
    \item The only vertex with the message before the first time unit is called the \emph{originator}.
    \item In each time unit, an informed vertex (\emph{sender}) can \emph{call} at most one of its uninformed neighbors (\emph{receiver}).
    \item During each time unit, all calls are performed in parallel.
    \item The process halts as soon as all the vertices in the graph are informed.
\end{enumerate}

Each call in the broadcasting process can be defined as $(u, v)$ where $u$ is the sender and $v$ is the receiver. The \emph{broadcast scheme} is represented by an ordered list containing sets of calls made in each time unit throughout the broadcast.
\begin{definition}\label{broadcast-time-definition}
    The broadcast time of a vertex $v$ in a given graph $G$ is the minimum number of time units required to broadcast in $G$ if $v$ is the originator and is denoted by $b(v, G)$. The broadcast time of a given graph $G$ is the maximum broadcast time from any originator in $G$, formally $b(G) = max_{v \in V(G)}\{b(v,G)\}$.
\end{definition}
The general broadcast time decision problem is formally defined as follows. 
\begin{definition}\label{broadcast-time-decision-definition}
Given a graph $G = (V, E)$ with a specified set of vertices $V_0 \subseteq V$ and a positive integer $k$, is there a sequence $V_0,E_1,V_1,E_2,V_2,...,E_k,V_k$ where $V_i \subseteq V_{i+1} \subseteq V, E_{i+1} \subseteq E, (0 \leq i < k)$, for every $(u,v) \in E_i, u \in V_{i-1}, v \in V_i \setminus V_{i-1}$ and $V_k = V$. For every $v \in V_i \setminus V_{i-1}$, there exists an edge $(u, v) \in E_i$. For every $u \in V_{i-1}$ and $v_1,v_2 \in V_i$, $\{(u,v_1),(u,v_2)\} \subseteq E_i \implies v_1=v_2$. Here, $k$ is the total broadcast time, $V_i$ is the set of informed vertices at round $i$, and $E_i$ is the set of edges used at round $i$. Clearly, when $|V_0|= 1$, this problem becomes our broadcast problem of determining $b(v, G)$ for an arbitrary vertex $v$ in an arbitrary graph $G$. 
\end{definition}

In this paper,  rather than the standard notation $b(v, G)$, we denote the optimal broadcast time of a vertex $v$ in graph $G$ as $t_{opt}(v)$, and the actual broadcast time under a broadcast scheme generated by our algorithm as $t_A(v)$. This avoids any confusion caused by referring to the broadcast time under our algorithm as $b(v, G)$, since in general, it is not optimal. $G$ is not included in the notation as it can be inferred from the surrounding context.

The broadcast time decision problem is NP-complete in an arbitrary graph \cite{computers-and-intractability,broadcasting-in-trees}. Furthermore, the problem was shown to be NP-complete in some restricted families of graphs, such as 3-regular planar graphs \cite{3-planar-graph-broadcasting}. The study of the parameterized complexity of the broadcast time problem was initiated in \cite{parameterized-complexity-general-tcs} and has progressed in the following years \cite{parameterized-complexity-general,parameterized-complexity-tw-2}. Moreover, there are few classes of graphs for which an exact polynomial-time algorithm is known. For example, exact linear-time algorithms exist for: trees \cite{broadcasting-in-trees-A.Prosk,broadcasting-in-trees}, unicyclic connected graphs \cite{linear-unicyclic-broadcasting,broadcasting-unicyclic-graphs}, necklace graphs (chain of rings) \cite{chains-of-rings}, and Harary-like graphs \cite{efficient-harary-like-graphs,harary-like-graphs}.

The intractability of the broadcast time problem has led to significant work on approximation algorithms. Ravi devised a polynomial-time ($O(\frac{\log^2n}{\log(\log n)})$)-approximation algorithm for the general broadcast time problem in \cite{ravi-approximation}. This approximation ratio was improved to $O(\log n)$ in \cite{bar-noy} by Bar-Noy et al. Subsequently, Elkin and Kortsarz again improved the approximation ratio to the current best known approximation in general graphs of $O(\frac{\log n}{\log(\log n)})$ \cite{logn-general-approximation,best-general-approximation}. Additionally, the general broadcast time problem was proved to be $\frac{3}{2}$-inapproximable in \cite{inapproximable-lb}. This was increased to a $(3-\epsilon)$-inapproximability result in \cite{logn-general-approximation}. Due to the significant advancements in the approximation of the general broadcast time problem, attention has been drawn to more constrained classes of graphs. 

Cactus graphs are a class of connected graphs in which any two simple cycles share at most one vertex. They have been a significant subject of interest in the field of graph algorithms as they are a simple class of graphs for which NP-hard problems sometimes have polynomial solutions. An exact algorithm to solve the broadcast time problem in cactus graphs was given in \cite{polynomial-k-restricted-cactus-broadcast}. The complexity of the algorithm on a cactus graph $G=(V, E)$ is $O(kk! \cdot m+n)$ where $m=|E|$, $n=|V|$, and $k$ is the maximum number of cycles that share a vertex. In a \emph{$k$-restricted cactus graph}, where $k$ is held constant, the algorithm is polynomial. This paper focuses on a very simple subfamily of cactus graphs, called \emph{k-cycle graphs} or \emph{flower graphs}, which are formed by merging $k$ cycles at a single vertex.

Very recent results have progressed through ever-more-constrained families of cactus graphs resulting in the determination of the complexity of the broadcast time problem in $k$-cycle graphs. First, Aminian et al. showed that the broadcast time problem is NP-complete in so-called "snowflake graphs" \cite{cactus-telephone-broadcasting}. Snowflake graphs are a subclass of cactus graphs that have a pathwidth of at most 2. Thus, they showed that the broadcast time problem is NP-complete in both general cactus graphs and general graphs with pathwidth at most 2. See Figure~\ref{fig:snowflake-graph-example} for an example of a snowflake graph.

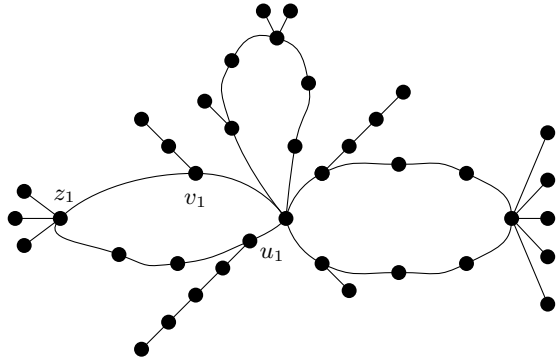
\begin{figure}[htb]
    \centering
    
    \begin{tikzpicture}[scale=1.2, every node/.style={circle, fill=black, inner sep=2pt}]
    
    \node (s) at (0,0) {};
    
    \node (v1) at (-1, 0.5) {};
    \node[fill=none] () at (-1, 0.2) {$v_1$};
    \node (z1) at (-2.5, 0) {};
    \node[fill=none] () at (-2.45, 0.25) {$z_1$};
    \node (c1_2) at (-1.85, -0.4) {};
    \node (c1_3) at (-1.2, -0.5) {};
    \node (u1) at (-0.4, -0.25) {};
    \node[fill=none] () at (-0.15, -0.4) {$u_1$};
    \draw[solid] plot [smooth, tension=1] coordinates { (s) (v1) (z1) (c1_2) (c1_3) (u1) (s) };

    \node (z1_a) at (-2.9, 0.3) {};
    \node (z1_b) at (-3, 0) {};
    \node (z1_c) at (-2.9, -0.3) {};
    \draw[solid] plot [smooth, tension=1] coordinates { (z1) (z1_a) };
    \draw[solid] plot [smooth, tension=1] coordinates { (z1) (z1_b) };
    \draw[solid] plot [smooth, tension=1] coordinates { (z1) (z1_c) };

    \node (v1_1) at (-1.3, 0.8) {};
    \node (v1_2) at (-1.6, 1.1) {};
    \draw[solid] plot [smooth, tension=1] coordinates { (v1) (v1_1) };
    \draw[solid] plot [smooth, tension=1] coordinates { (v1_1) (v1_2) };

    \node (u1_1) at (-0.7, -0.55) {};
    \node (u1_2) at (-1, -0.85) {};
    \node (u1_3) at (-1.3, -1.15) {};
    \node (u1_4) at (-1.6, -1.45) {};
    \draw[solid] plot [smooth, tension=1] coordinates { (u1) (u1_1) };
    \draw[solid] plot [smooth, tension=1] coordinates { (u1_1) (u1_2) };
    \draw[solid] plot [smooth, tension=1] coordinates { (u1_2) (u1_3) };
    \draw[solid] plot [smooth, tension=1] coordinates { (u1_3) (u1_4) };

    \node (v2) at (0.4, 0.5) {};
    \node (v2) at (0.4, 0.5) {};
    \node (c2_1) at (1.25, 0.6) {};
    \node (c2_2) at (2, 0.5) {};
    \node (z2) at (2.5, 0) {};
    \node (c2_4) at (2, -0.5) {};
    \node (c2_5) at (1.25, -0.6) {};
    \node (u2) at (0.4, -0.5) {};
    \draw[solid] plot [smooth, tension=1] coordinates { (s) (v2) (c2_1) (c2_2) (z2) (c2_4) (c2_5) (u2) (s) };

    \node (z2_1) at (2.9, 0.95) {};
    \node (z2_2) at (2.9, 0.425) {};
    \node (z2_3) at (2.9, 0) {};
    \node (z2_4) at (2.9, -0.425) {};
    \node (z2_5) at (2.9, -0.95) {};
    \draw[solid] plot [smooth, tension=1] coordinates { (z2) (z2_1) };
    \draw[solid] plot [smooth, tension=1] coordinates { (z2) (z2_2) };
    \draw[solid] plot [smooth, tension=1] coordinates { (z2) (z2_3) };
    \draw[solid] plot [smooth, tension=1] coordinates { (z2) (z2_4) };
    \draw[solid] plot [smooth, tension=1] coordinates { (z2) (z2_5) };

    \node (v2_1) at (0.7, 0.8) {};
    \node (v2_2) at (1, 1.1) {};
    \node (v2_3) at (1.3, 1.4) {};
    \draw[solid] plot [smooth, tension=1] coordinates { (v2) (v2_1) };
    \draw[solid] plot [smooth, tension=1] coordinates { (v2_1) (v2_2) };
    \draw[solid] plot [smooth, tension=1] coordinates { (v2_2) (v2_3) };

    \node (u2_1) at (0.7, -0.8) {};
    \draw[solid] plot [smooth, tension=1] coordinates { (u2) (u2_1) };

    \node (v3) at (-0.6, 1) {};
    \node (c3_2) at (-0.6, 1.75) {};
    \node (z3) at (-0.1, 2) {};
    \node (c3_5) at (0.25, 1.5) {};
    \node (u3) at (0.1, 0.8) {};
    \draw[solid] plot [smooth, tension=1] coordinates { (s) (v3) (c3_2) (z3) (c3_5) (u3) (s) };

    \node (z3_1) at (-0.25, 2.3) {};
    \node (z3_2) at (0.05, 2.3) {};
    
    \draw[solid] plot [smooth, tension=1] coordinates { (z3) (z3_1) };
    \draw[solid] plot [smooth, tension=1] coordinates { (z3) (z3_2) };

    \node (v3_1) at (-0.9, 1.3) {};
    \draw[solid] plot [smooth, tension=1] coordinates { (v3) (v3_1) };

    
    \end{tikzpicture}
    
    \caption{An example of a snowflake graph: a graph consisting of a $k$-cycle graph with paths originating from the vertices adjacent to the central vertex. One vertex on each cycle can also have any number of additional adjacent vertices. Snowflake graphs are cactus graphs with pathwidth at most 2 \cite{cactus-telephone-broadcasting}.} 
    \label{fig:snowflake-graph-example}
\end{figure}

Next, Egami et al. proved that the broadcast time problem is NP-complete in cactus graphs that are vertex deletion distance 1 to a path forest \cite{broadcasting-struct-restrictions} (see Figure~\ref{fig:structural-restrictions-figure}). The reduction was performed from a restricted version of the numerical matching with target sums problem (NMTS), which simplified the reduction in \cite{cactus-telephone-broadcasting}. Figure~\ref{fig:structural-restrictions-figure} is an example of a graph resulting from their reduction. This class is the set of $k$-cycle graphs with paths of arbitrary length originating from the central vertex.

\begin{figure}[htb]
    \centering
    
    \begin{tikzpicture}[scale=1.1, every node/.style={circle, fill=black, inner sep=2pt}]
    
    \node (s) at (0,0) {};
    
    \node (v1) at (-1, 1) {};
    \node (c1_1) at (-1.7, 0.5) {};
    \node (c1_2) at (-2, 0) {};
    \node (c1_3) at (-1.7, -0.5) {};
    \node (u1) at (-1, -1) {};
    \draw[solid] plot [smooth, tension=1] coordinates { (s) (v1) (c1_1) (c1_2) (c1_3) (u1) (s) };

    \node (v2) at (0.8, 1) {};
    \node (c2_0) at (1.35, 1.25) {};
    \node (c2_1) at (1.7, 1) {};
    \node (c2_2) at (2, 0.6) {};
    \node (c2_3) at (2.2, 0) {};
    \node (c2_4) at (2, -0.6) {};
    \node (c2_5) at (1.7, -1) {};
    \node (c2_6) at (1.4, -1.2) {};
    \node (u2) at (0.8, -1) {};
    \draw[solid] plot [smooth, tension=1] coordinates { (s) (v2) (c2_0) (c2_1) (c2_2) (c2_3) (c2_4) (c2_5) (c2_6) (u2) (s) };

    \node (Q1_1) at (-0.2, 0.7) {};
    \node (Q1_2) at (-0.4, 1) {};
    \node (Q1_3) at (-0.6, 1.3) {};
    \node (Q1_4) at (-0.8, 1.6) {};
    \node (Q1_5) at (-1, 1.9) {};
    \node (Q1_6) at (-1.2, 2.2) {};
    \node (q1) at (-1.4, 2.5) {};
    \draw[solid] plot [smooth, tension=1] coordinates { (s) (Q1_1)  (Q1_2) (Q1_3) (Q1_4) (Q1_5) (Q1_6) (q1)};
    
    \node (Q2_2) at (0, 1.0) {};
    \node (Q2_3) at (0, 1.3) {};
    \node (Q2_4) at (0, 1.6) {};
    \node (Q2_5) at (0, 1.9) {};
    \node (Q2_6) at (0, 2.2) {};
    \node (q2) at (0, 2.5) {};
    \draw[solid] plot [smooth, tension=1] coordinates { (s)  (Q2_2) (Q2_3) (Q2_4) (Q2_5) (Q2_6) (q2)};
    
    \node (qt) at (0, -1) {};
    \draw[solid] plot [smooth, tension=1] coordinates { (s) (qt)};
    
    
    \end{tikzpicture}
    
    \caption{An example of a graph created by the reduction in \cite{broadcasting-struct-restrictions}. The graph consists of a k-cycle graph with paths connected to the central vertex.}
    \label{fig:structural-restrictions-figure}
\end{figure}
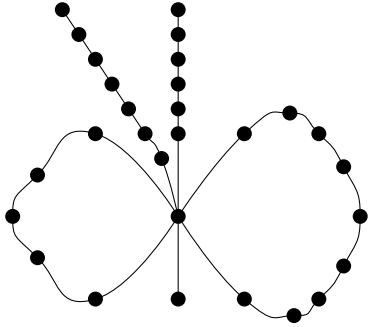

Finally, the broadcast time problem was shown to be NP-complete in $k$-cycle graphs in \cite{bringolf2026hardnessapproximationbroadcastingstructured}. The reduction is similar to the one in \cite{broadcasting-struct-restrictions}, but originates from a more restricted version of NMTS to eliminate the need for paths starting at the central vertex. The same paper also introduces a Polynomial-Time Approximation Scheme (PTAS) for the broadcast time problem in $k$-cycle graphs.  

In contrast, this paper presents a simple linear-time $(1.5-\epsilon$)-approximation algorithm for the broadcast time problem in $k$-cycle graphs. By $(1.5-\epsilon)$, we mean an approximation ratio strictly less than $1.5$. In other words, for any given graph, the algorithm's approximation ratio is $1.5-\epsilon$, where $\epsilon > 0$ depends on the structure of the graph. Although the PTAS mentioned above can achieve a better approximation, reaching a 1.5-approximation would incur a runtime complexity of $O(2^{64}n^2\log{n}) \approx O(10^{19}n^2\log{n})$, where $n$ is the number of vertices in the graph. In comparison, our algorithm achieves this approximation with a runtime complexity of $2n$. This, along with its simplicity, makes our algorithm a very practical and relevant option for implementation. It also improves the previously best-known fixed-ratio approximation guarantee of 2, presented in \cite{k-cycle-2-approximation}.

The paper is organized as follows. In Section~\ref{sec:k-cycle-graphs}, we describe the family of $k$-cycle graphs and discuss previous related findings. In Section~\ref{sec:simple-k-cycle}, we present our approximation algorithm. In Sections \ref{sec:broadcasting-case1} and \ref{sec:broadcasting-case2}, we provide proofs of our proposed algorithm's approximation ratio. In Section~\ref{sec:particular-cases}, we present infinite graph classes for which our algorithm is exact or approaches the approximation ratio. Finally, we conclude the paper in Section~\ref{sec:conclusion}.
\section{k-Cycle Graphs}
\label{sec:k-cycle-graphs}

In this section, we discuss a subfamily of cactus graphs that are usually referred to as \emph{$k$-cycle} graphs or \emph{flower graphs}. A $k$-cycle graph $G = (C_1,C_2,...,C_k)$ is obtained by connecting $k$ cycles of arbitrary lengths at a single vertex $v_c$, called the central vertex.

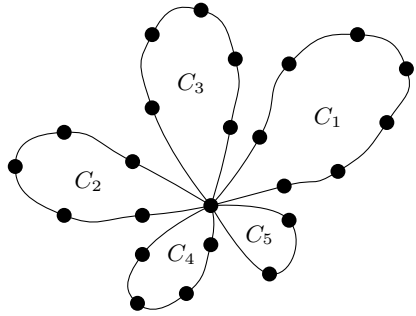
\begin{figure}[htb] 
    \centering
    
    \begin{tikzpicture}[scale=1.3, every node/.style={circle, fill=black, inner sep=2pt}]
    
    \node (s) at (0,0) {};

    \node[fill=none] (v_c) at (-0.05, 0.25) {$v_c$};
    
    \node (v1) at (-0.7, -0.5) {};
    \node (z1) at (-0.75, -1) {};
    \node (c1_2) at (-0.25, -0.9) {};
    \node (u1) at (0, -0.4) {};
    \draw[solid] plot [smooth, tension=1] coordinates { (s) (v1) (z1) (c1_2) (u1) (s) };
    \node[fill=none] (c4_name) at (-0.3, -0.5) {$C_4$};

    \node (v2) at (0.8, -0.15) {};
    \node (u2) at (0.6, -0.7) {};
    \draw[solid] plot [smooth, tension=1] coordinates { (s) (v2) (u2) (s) };
    \node[fill=none] (c3_name) at (0.5, -0.3) {$C_5$};
    
    \node (v3) at (-0.6, 1) {};
    \node (c3_2) at (-0.6, 1.75) {};
    \node (z3) at (-0.1, 2) {};
    \node (c3_5) at (0.25, 1.5) {};
    \node (u3) at (0.2, 0.8) {};
    \draw[solid] plot [smooth, tension=1] coordinates { (s) (v3) (c3_2) (z3) (c3_5) (u3) (s) };
    \node[fill=none] (c3_name) at (-0.2, 1.25) {$C_3$};

    \node (c4_1) at (-0.8, 0.45) {};
    \node (c4_2) at (-1.5, 0.75) {};
    \node (c4_3) at (-2, 0.4) {};
    \node (c4_4) at (-1.5, -0.1) {};
    \node (c4_5) at (-0.7, -0.1) {};
    \draw[solid] plot [smooth, tension=1] coordinates { (s) (c4_1) (c4_2) (c4_3) (c4_4) (c4_5) (s)};
    \node[fill=none] (c2_name) at (-1.25, 0.25) {$C_2$};

    \node (c4_0) at (0.5, 0.7) {};
    \node (c4_1) at (0.8, 1.45) {};
    \node (c4_2) at (1.5, 1.75) {};
    \node (c4_3) at (2, 1.4) {};
    \node (c4_4) at (1.8, 0.85) {};
    \node (c4_5) at (1.3, 0.35) {};
    \node (c4_6) at (0.75, 0.2) {};
    \draw[solid] plot [smooth, tension=1] coordinates { (s) (c4_0) (c4_1) (c4_2) (c4_3) (c4_4) (c4_5) (c4_6) (s)};
    \node[fill=none] (c1_name) at (1.2, 0.9) {$C_1$};

    
    \end{tikzpicture}
    
    \caption{An example of a $k$-cycle graph with $k=5$ and $l_1=7, l_2=5, l_3=5, l_4=4, l_5=2$.}
\end{figure}

We assume that cycles are indexed in non-increasing order of their lengths. Formally, $l_1 \geq l_2 \geq ... \geq l_k \geq 2$, where $l_i$ is the number of vertices on cycle $C_i$ (excluding $v_c$) for all $1 \leq i \leq k$.

As noted above, the broadcast time problem has been researched at length for cactus graphs and graphs with pathwidth at most 2 in \cite{cactus-telephone-broadcasting}. The authors also present a polynomial-time 2-approximation algorithm in cactus graphs that finds exact broadcast times using the 2-broadcasting model; a model where vertices can make 2 calls per time unit. They exploit the structure of cactus graphs to recursively compute the broadcast times of disjoint connected components induced by the removal of informed vertices. This extends the previous best approximation ratio of 2 in $k$-cycle graphs found in \cite{k-cycle-2-approximation} to general cactus graphs. 

The main algorithm introduced in \cite{k-cycle-2-approximation} ($S_{\text{cycle}}$) maintains 2 lists of cycles sorted in non-increasing order of the number of uninformed vertices. The lists contain all cycles that have been called once, or that have not yet been called, respectively. In each time unit, it compares the number of uninformed vertices in the first cycle of each list to determine which to inform. In this paper, we design a simple algorithm that achieves a better approximation considering only the lengths of cycles. In \cite{k-cycle-2-approximation}, the authors also presented the following lower bounds for the broadcast time problem in $k$-cycle graphs, which we use in our proof.
\begin{lemma}\label{lma:lower_bounds_lemma}
(\cite{k-cycle-2-approximation}). Let $G_k$ be a $k$-cycle graph where the originator is the central vertex $u$. Then 
\begin{enumerate}
    \item $b(u) \geq k+1$ 
    \item $b(u) \geq \ceil*{\frac{l_j+2j-1}{2}}$, $1 \leq j \leq k$.
\end{enumerate}
\end{lemma}
\begin{lemma}\label{lma:d_gt_0_lowerbound}
    (\cite{k-cycle-2-approximation}) Let $G_k$ be a $k$-cycle graph where the originator is any vertex $w$ on a cycle $C_m$ and the length of the shortest path from $w$ to vertex $u$ is $d$. Then $b(w) \geq d + \ceil*{\frac{l_j+2j-2}{2}}$, $1 \leq j \leq k$, where $j \neq m$.
\end{lemma}

\section{\textsc{Simple-k-cycle} Algorithm}\label{sec:simple-k-cycle}

In this section, we will describe \textsc{Simple-k-cycle}: our approximation algorithm for the broadcast time problem in $k$-cycle graphs (see Algorithm~\ref{alg:simple-k-cycle}). Unlike the 2-approximation algorithm and PTAS described above, our algorithm is very simple. Recall, as stated in the previous section, that by convention, the cycles of our $k$-cycle graph are indexed in non-increasing order of cycle length. Formally, $l_1 \geq l_2 \geq ... \geq l_k \geq 2$. This convention is used implicitly throughout our algorithm and approximation ratio proofs. Our algorithm has two distinct cases: the originator vertex ($u$) is the central vertex $v_c$, or the originator vertex is found on a cycle $C_m$ and is $d \geq 1$ vertices away from $v_c$. When $u=v_c$, $v_c$ calls each cycle $C_i$ in time units $i$ and $i+k$ for all $i$, $1 \leq i \leq k$. 

When $u \neq v_c$ and is found on $C_m$, $d$ vertices away from $v_c$, the algorithm is as follows: $u$ informs its neighbour along the shortest path to $v_c$ in time unit 1. In time unit 2, $u$ informs its other neighbour. In time unit  $d$, $v_c$ is informed by its neighbour on $C_m$. Then, in time units $d+1, ..., d+m-1$, $v_c$ calls cycles $C_1, ..., C_{m-1}$. Then, in time units $d+m, ..., d+k-1$, $v_c$ calls cycles $C_{m+1}, ..., C_k$. $v_c$ calls each cycle $C_i$ for a second time in time unit $d+k-1+i$ for all $i=1, ..., k$.

\begin{algorithm}[htb]
\caption{SIMPLE-K-CYCLE}\label{alg:simple-k-cycle}
\hspace*{\algorithmicindent} \textbf{Input} A $k$-cycle graph $G = \{C_1, ..., C_k\}$ with central vertex $v_c$ and an originator vertex $u$ \\ 
\hspace*{\algorithmicindent} \textbf{Output} A broadcast scheme with broadcast time $t_A(v_c, G)$
\begin{algorithmic}
\Procedure{Simple-K-Cycle}{$G$, $u$}
    \If{$u = v_c$}
        \For {$1 \leq i \leq k$}
            \State $v_c$ calls $C_i$ in time unit $i$
        \EndFor
        \For {$k+1 \leq i \leq 2k$}
            \If{$C_{i-k}$ has any uninformed vertices}
                \State $v_c$ calls $C_{i-k}$ in time unit $i$
            \EndIf
        \EndFor
    \Else
         \State In time unit 1, $u$ informs its neighbour along the shorter path to $v_c$
        \State In time unit 2, $u$ informs its neighbour along the longer path to $v_c$
        \State $v_c$ is informed by its neighbour on $C_m$ in time unit $d$

        \For{$1 \leq i \leq m-1$}
            \State $v_c$ calls $C_i$ in time unit $d+i$
        \EndFor
        \For{$m+1 \leq i \leq k$}
            \State $v_c$ calls $C_i$ in time unit $d+i-1$
        \EndFor
        \For{$1 \leq i \leq k$}
            \State $v_c$ calls $C_i$ in time unit $d+k-1+i$
        \EndFor
    \EndIf
    \State

\EndProcedure
\end{algorithmic}
\end{algorithm}

\paragraph*{Complexity Analysis}
The complexity of \textsc{Simple-K-Cycle} is in $O(k)$. $k<n$, where $n=|V|$, so the complexity of the algorithm is also in $O(n)$. The algorithm iterates over each of the $k$ cycles at most twice. In each iteration, the current cycle is called by $v_c$. Thus, once we have iterated over each cycle twice ($2k$ iterations), we have completely defined our broadcast scheme. Once we have determined our broadcast scheme, we can trivially compute the time unit in which each cycle is fully informed. Specifically, the time unit ($t_i$) in which any particular cycle $C_i$ is fully informed\textemdash if it is called in time units $t_1$ and $t_2$ with $t_1<t_2$\textemdash is calculated as:
    \begin{equation*}
        t_i = \begin{cases}
        t_2-1+\ceil*{\frac{l_i-t_2+t_1}{2}} & C_i \text{ receives 2 calls}\\
        t_1-1+l_i & C_i \text{ receives 1 call}
        \end{cases}
    \end{equation*}
We then require $k-1$ comparisons to determine the last cycle to be fully informed, which gives us the broadcast time of our graph under our approximation algorithm. So, the complexity of \textsc{Simple-K-Cycle} is in $O(k)$ and consequently, is also in $O(n)$. To be precise, the algorithm performs $4k-1$ operations to compute the broadcast time ($2k$ iterations to determine the broadcast scheme, $k$ iterations to compute all $t_i$, and $k-1$ comparisons). Since each cycle must have at least 2 vertices (not including $v_c$), $2k<n$, so $k < \frac{n}{2}$, so our algorithm performs less than $2n$ operations.

\section{Broadcasting from the Central Vertex}
\label{sec:broadcasting-case1}

This section discusses the broadcast time generated by \textsc{Simple-K-Cycle} when the originator is $v_c$. Recall, the algorithm informs cycle $C_i$, $1 \leq i \leq k$ in time units $i$ and $k+i$. See Figure~\ref{fig:case1-example}.

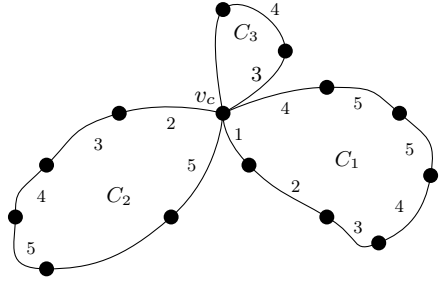
\begin{figure}[htb]
    \centering
    
    \begin{tikzpicture}[scale=1.3, every node/.style={circle, fill=black, inner sep=2pt}]
    
    \node (s) at (0,0) {};
    \node[fill=none] at (-0.17,0.15) {$v_c$};
    
    \node (v1) at (-1, 0) {};
    \node (c1_1) at (-1.7, -0.5) {};
    \node (c1_2) at (-2, -1) {};
    \node (c1_3) at (-1.7, -1.5) {};
    \node (u1) at (-0.5, -1) {};
    \draw[solid] plot [smooth, tension=1] coordinates { (s) (v1) (c1_1) (c1_2) (c1_3) (u1) (s) };

    \node[fill=none, scale=0.9] () at (1.2, -0.45) {$C_1$};
    \node (v2) at (1, 0.25) {};
    \node (c2_1) at (1.7, 0) {};
    \node (c2_2) at (2, -0.6) {};
    \node (c2_4) at (1.5, -1.25) {};
    \node (c2_5) at (1, -1) {};
    \node (u2) at (0.25, -0.5) {};
    \draw[solid] plot [smooth, tension=1] coordinates { (s) (v2) (c2_1) (c2_2) (c2_4) (c2_5) (u2) (s) };

    \node (v2) at (0.6, 0.6) {};
    \node (c2_1) at (0, 1) {};
    \draw[solid] plot [smooth, tension=1] coordinates { (s) (v2) (c2_1) (s) };
    \node[fill=none, scale=0.9] () at (0.225, 0.775) {$C_3$};

    \node[fill=none, scale=0.75] (v2) at (-0.5, -0.1) {2};
    \node[fill=none, scale=0.75] (v2) at (-1.2, -0.3) {3};
    \node[fill=none, scale=0.75] (v2) at (-1.75, -0.8) {4};
    \node[fill=none, scale=0.75] (v2) at (-1.85, -1.3) {5};
    \node[fill=none, scale=0.75] (v2) at (-0.3, -0.5) {5};
    \node[fill=none, scale=0.9] () at (-1, -0.8) {$C_2$};

    \node[fill=none, scale=0.9] (v2) at (0.325, 0.375) {3};
    \node[fill=none, scale=0.75] (v2) at (0.5, 1) {4};

    \node[fill=none, scale=0.75] (v2) at (0.6, 0.05) {4};
    \node[fill=none, scale=0.75] (v2) at (1.3, 0.1) {5};
    \node[fill=none, scale=0.75] (v2) at (1.8, -0.35) {5};
    \node[fill=none, scale=0.75] (v2) at (1.7, -0.9) {4};
    \node[fill=none, scale=0.75] (v2) at (1.3, -1.1) {3};
    \node[fill=none, scale=0.75] (v2) at (0.7, -0.7) {2};
    \node[fill=none, scale=0.75] (v2) at (0.15, -0.2) {1};
    
    \end{tikzpicture}
    
    \caption{An example of a broadcast scheme generated by \textsc{Simple-K-Cycle} on a $k$-cycle graph with $l_1=6, l_2=5, l_3=2$ and originator $v_c$. $v_c$ calls cycle $C_1$ in time units 1 and 4, cycle $C_2$ in time units 2 and 5, and cycle $C_3$ in time unit 3. The broadcast is completed in 5 time units.}
    \label{fig:case1-example}
\end{figure}

\begin{theorem}\label{th:case1}
    Algorithm~\ref{alg:simple-k-cycle} is a polynomial-time $(1.5-\epsilon)$-approximation algorithm for general $k$-cycle graphs when the originator is $v_c$.
\end{theorem}

\begin{proof}
    Let $t_A(v_c)$ be the broadcast time generated by Algorithm~\ref{alg:simple-k-cycle}, and $t_{opt}(v_c)$ be the optimal broadcast time for $G$ with originator $v_c$. To calculate $\frac{t_A(v_c)}{t_{opt}(v_c)}$, we will consider two cases: $t_A(v_c) \geq 2k$, and $t_A(v_c) < 2k$.
    \paragraph*{Case 1 ($t_A(v_c) \geq 2k$):}
    If $t_A(v_c) \geq 2k$, then all cycles $C_1, ..., C_k$ get 2 calls from $v_c$. Since cycle $C_i$ is informed by $v_c$ at times $i$ and $k+i$, $k$ vertices in $C_i$ are informed by time $k+i$. Thus, starting at time $k+i$, there are $\max\{0,l_i-k\}$ uninformed vertices in cycle $C_i$. In each time unit starting at time $k+i$, if $C_i$ has any remaining uninformed vertices, two additional vertices are informed in $C_i$. So, it takes $\max\{0, \ceil*{\frac{l_i-k}{2}}\}$ time units to inform the remaining vertices in cycle $C_i$ once it receives the second call from $v_c$ at time unit $k+i$.
    So, for all $i, 1 \leq i \leq k$, such that $l_i \geq k$, cycle $C_i$ is fully informed by time: \\
    \begin{equation} \label{eq:t_i_from_center}
        t_i =k+i-1 + \ceil*{\frac{l_i-k}{2}} \leq \ceil*{\frac{2i-2+k+l_i}{2}}
    \end{equation}
Also, from the second lower bound from Lemma~\ref{lma:lower_bounds_lemma}, for any $1 \leq i \leq k$, we get
    \begin{equation} \label{eq:l_i_lowerbound}
        2t_{opt}(v_c) - 2i + 1 \geq l_i
    \end{equation}
Naturally, $t_A(v_c) = \max_{1\leq i \leq k}\{t_i\}$. By our assumption that $t_A(v_c) \geq 2k$, we know that the final cycle to be fully informed (the cycle $C_i$ with the maximum $t_i$) must be called by $v_c$ twice and thus must have $l_i \geq k$. Combining (\ref{eq:t_i_from_center}), and (\ref{eq:l_i_lowerbound}), we get the following inequality where $i$ is the index that maximizes $\max_{1 \leq i\leq k}\{t_i\}$.
    \begin{align*}
       &t_i = \ceil*{\frac{2i-2+k+l_i}{2}} \leq \ceil*{\frac{2i-2+k+2t_{opt}(v_c)-2i+1}{2}} \\
       &t_i \leq t_{opt}(v_c) + \ceil*{\frac{k-1}{2}}
       \label{eq:t_i_lb_case1}
    \end{align*}
        Additionally from the first lower bound of Lemma~\ref{lma:lower_bounds_lemma}, we have $t_{opt}(v_c)-1 \geq k$. So, we know that $t_A(v_c) = \max_{1 \leq i \leq k}\{t_i\} \leq t_{opt}(v_c) + \ceil*{\frac{t_{opt}(v_c)-2}{2}}$. As a result,
    \begin{align*}
        &\frac{t_A(v_c)}{t_{opt}(v_c)} \leq \frac{t_{opt}(v_c) + \ceil*{\frac{t_{opt}(v_c)-2}{2}}}{t_{opt}(v_c)} \leq 1 + \frac{\ceil*{\frac{t_{opt}(v_c)-2}{2}}}{t_{opt}(v_c)} =  1 + \frac{\ceil*{\frac{t_{opt}(v_c)}{2}}}{t_{opt}(v_c)} - \frac{1}{t_{opt}(v_c)}\\
        &\leq 1 + \frac{0.5t_{opt}(v_c)+0.5}{t_{opt}(v_c)} - \frac{1}{t_{opt}(v_c)} = 1.5 - \frac{0.5}{t_{opt}(v_c)} < 1.5
    \end{align*}
    So, $\frac{t_A(v_c)}{t_{opt}(v_c)} < 1.5$ when $t_A(v_c) \geq 2k$
    \paragraph*{Case 2 ($t_A(v_c) < 2k$):}
     We will prove this case by contradiction. If $t_A(v_c) = k+p$, for some $1 \leq p < k$, then cycles $C_1, ..., C_p$ are informed by $v_c$ twice, while cycles $C_{p+1}, ..., C_k$ are only informed once. Cycles $C_1, ..., C_p$ are informed exactly like in case 1, so they are fully informed before $1.5t_{opt}(v_c)$. Consider cycles $C_{p+1}, ..., C_k$.

    Let $i$ be the smallest value $p+1 \leq i \leq k$ such that cycle $C_i$ is only fully informed at time unit $t_A(v_c)=p+k$. For the sake of contradiction, assume $t_A(v_c)=k+p \geq 1.5t_{opt}(v_c)$. In this case, $C_i$ is fully informed at exactly $t_A(v_c)$, so $t_i=t_A(v_c)$. Additionally, since cycle $C_i$ is informed only once at time $i$, we have $t_A (v_c)= i-1+l_i$. From this and (\ref{eq:l_i_lowerbound}),    
    \begin{align*}
        &t_A(v_c)=p+k=i-1+l_i\\
        &t_A(v_c)=p+k\leq i-1+2t_{opt}(v_c)-2i+1\\
        &t_A(v_c)=p+k\leq 2t_{opt}(v_c)-i
    \end{align*}
    From our assumption that $k+p \geq 1.5t_{opt}(v_c)$ and above,
    \begin{align*}
        &1.5t_{opt}(v_c) \leq p+k \leq 2t_{opt}(v_c)-i
    \end{align*}
    Then, $1.5t_{opt}(v_c) \leq 2t_{opt}(v_c)-i$ which leads to $i \leq \frac{t_{opt}(v_c)}{2}$. However, from our selection of $i$, we know that $p+1 \leq i$. Transitively, this means that $p+1 \leq \frac{t_{opt}(v_c)}{2}$, so $ p \leq \frac{t_{opt}(v_c)}{2}-1$. Also, since $t_{opt}(v_c) \geq k+1$, we have that $k \leq t_{opt}(v_c)-1$. Thus, substituting $p$ and $k$ from above,
    \begin{equation*}
        t_A(v_c) = p+k \leq \frac{t_{opt}(v_c)}{2}-1+t_{opt}(v_c)-1 = 1.5t_{opt}(v_c)-2
    \end{equation*}
    This is a contradiction with the assumption that $t_A(v_c) \geq 1.5t_{opt}(v_c)$, so $\frac{t_A(v_c)}{t_{opt}(v_c)} < 1.5$ when $t_A(v_c) < 2k$. Therefore, Algorithm~\ref{alg:simple-k-cycle} is a $(1.5-\epsilon)$-approximation algorithm when $u=v_c$.
\end{proof}
\section{Broadcasting from a Cycle Vertex}\label{sec:broadcasting-case2}

This section discusses the broadcast time generated by \textsc{Simple-K-Cycle} when the originator is not $v_c$. Instead, in this section, we assume that the originator $u$ is found $d$ vertices away from $v_c$ on some cycle $C_m$. Recall, the algorithm informs along the shortest path to $v_c$. Once $v_c$ is informed, it proceeds as in the previous section, but skips the first call to $C_m$ since it is already informed. See Figure~\ref{fig:case2-example}.

\begin{figure}[htb] 
    \centering
    
    \begin{tikzpicture}[scale=1.3, every node/.style={circle, fill=black, inner sep=2pt}]
    
    \node (s) at (0,0) {};
    
    \node (v1) at (-0.7, -0.5) {};
    \node (z1) at (-0.75, -1) {};
    \node (c1_2) at (-0.25, -0.9) {};
    \node (u1) at (0, -0.4) {};
    \draw[solid] plot [smooth, tension=1] coordinates { (s) (v1) (z1) (c1_2) (u1) (s) };
    \node[fill=none, scale=0.8] (c4_name) at (-0.5, -0.75) {$C_3$};
    \node[fill=none, scale=0.8] (u) at (-0.35, -0.3) {4};
    \node[fill=none, scale=0.8] (u) at (-0.95, -0.75) {5};
    \node[fill=none, scale=0.8] (u) at (-0.5, -1.2) {6};
    \node[fill=none, scale=0.8] (u) at (-0.135, -0.6) {7};
    \node[fill=none, scale=0.8] (u) at (-0.025, -0.2) {8};

    \node (v2) at (1, -0.15) {};
    \node (u2) at (0.6, -0.9) {};
    \draw[solid] plot [smooth, tension=1] coordinates { (s) (v2) (u2) (s) };
    \node[fill=none, scale=0.8] (c3_name) at (0.6, -0.4) {$C_4$};
    \node[fill=none, scale=0.8] (c3_name) at (0.275, -0.75) {5};
    \node[fill=none, scale=0.8] (c3_name) at (1.1, -0.6) {6};

    \node (c4_1) at (-0.5, 0.65) {};
    \node (c4_1_1) at (-0.9, 0.9) {};
    \node (c4_2) at (-1.4, 0.95) {};
    \node (c4_2_2) at (-1.85, 0.8) {};
    \node (c4_3) at (-2, 0.4) {};
    \node (c4_4) at (-1.8, -0) {};
    \node (c4_4_2) at (-1.3, -0.2) {};
    \node (c4_5) at (-0.7, -0.1) {};
    \node (vc)[fill=none] at (0.2, -0.1) {$v_c$};
    \draw[solid] plot [smooth, tension=1] coordinates { (s) (c4_1) (c4_1_1) (c4_2) (c4_2_2) (c4_3) (c4_4) (c4_4_2) (c4_5) (s)};
    \node[fill=none, scale=0.8] (c2_name) at (-1.1, 0.4) {$C_2=C_m$};
    \node[fill=none] (u) at (-0.9, 1.1) {$u$};
    \node[fill=none, scale=0.8] (u) at (-0.7, 0.7) {1};
    \node[fill=none, scale=0.8] (u) at (-0.35, 0.35) {2};
     \node[fill=none, scale=0.8] (u) at (-1.15, 0.85) {2};
    \node[fill=none, scale=0.8] (u) at (-1.6, 0.82) {3};
    \node[fill=none, scale=0.8] (u) at (-1.9, 0.6) {4};
    \node[fill=none, scale=0.8] (u) at (-1.9, 0.2) {5};
    \node[fill=none, scale=0.8] (u) at (-1.55, -0.05) {6};
    \node[fill=none, scale=0.8] (u) at (-1.05, -0.05) {7};
    \node[fill=none, scale=0.8] (u) at (-0.4, 0.05) {7};

    \node (c4_0) at (0, 0.7) {};
    \node (c4_0_2) at (0.05, 1.1) {};
    \node (c4_1) at (0.3, 1.45) {};
    \node (c4_1_2) at (0.55, 1.7) {};
    \node (c4_2) at (1, 1.75) {};
    \node (c4_3) at (1.5, 1.4) {};
    \node (c4_4) at (1.5, 0.85) {};
    \node (c4_5) at (1.3, 0.35) {};
    \node (c4_6) at (0.75, 0.2) {};
    \draw[solid] plot [smooth, tension=1] coordinates { (s) (c4_0) (c4_0_2) (c4_1) (c4_1_2) (c4_2) (c4_3) (c4_4) (c4_5) (c4_6) (s)};
    \node[fill=none, scale=0.8] (c1_name) at (0.8, 0.9) {$C_1$};
    \node[fill=none, scale=0.8] (c1_name) at (0.3, 0.2) {3};
    \node[fill=none, scale=0.8] (c1_name) at (1, 0.35) {4};
    \node[fill=none, scale=0.8] (c1_name) at (1.4, 0.65) {5};
    \node[fill=none, scale=0.8] (c1_name) at (1.5, 1.15) {6};
    \node[fill=none, scale=0.8] (c1_name) at (1.275, 1.5) {7};
    \node[fill=none, scale=0.8] (c1_name) at (0.775, 1.65) {8};
    \node[fill=none, scale=0.8] (c1_name) at (0.075, 0.35) {6};
    \node[fill=none, scale=0.8] (c1_name) at (0.075, 0.9) {7};
    \node[fill=none, scale=0.8] (c1_name) at (0.225, 1.225) {8};

    
    \end{tikzpicture}
    
    \caption{An example of a broadcast scheme generated by \textsc{Simple-K-Cycle} on a $k$-cycle graph with $l_1=9, l_2=8, l_3=4, l_4=2$ and originator $u$ with $dist(u, v_c)=d=2$. $u$ informs towards $v_c$ in time unit 1, then informs its other neighbour in time unit 2. $v_c$ is informed in time unit 2. Then, $v_c$ calls cycle $C_1$ in time units 3 and 6, cycle $C_2$ in time unit 7, cycle $C_3$ in time units 4 and 8, and cycle $C_4$ in time unit 5. The broadcast is completed in 8 time units.}
    \label{fig:case2-example}
\end{figure}
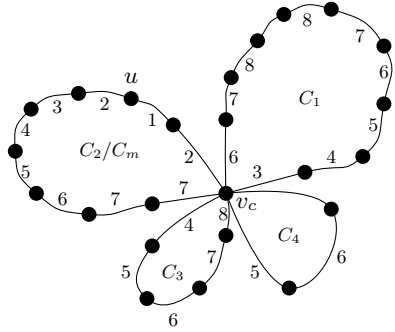

\begin{theorem}\label{th:case2}
    Algorithm~\ref{alg:simple-k-cycle} is a polynomial-time $(1.5-\epsilon)$-approximation algorithm for general $k$-cycle graphs when the originator $u$ is not $v_c$.
\end{theorem}

\begin{proof}
Since $u \neq v_c$, we have that $dist(u, v_c) = d \geq 1$. It is clear that all the vertices on cycle $C_m$ are informed no later than when the originator is $v_c$. Additionally, we know that $t_{opt}(u) \geq t_{opt}(v_c)$. This is because a broadcast scheme originating at $v_c$ can fully inform the cycle containing $u$ in the same number of time units as a broadcast scheme originating at $u$. Moreover, all other cycles can be fully informed in a broadcast scheme originating at $v_c$ at least as quickly as in a broadcast scheme originating at $u$ due to the additional time it takes to inform $v_c$ in the latter. So, the optimal broadcast time originating at $u$ must not be less than the optimal broadcast time originating at $v_c$. So, by Theorem~\ref{th:case1}, cycle $C_m$ is fully informed in less than $1.5t_{opt}(v_c) \leq 1.5t_{opt}(u)$ time units. We will therefore only consider the case where the broadcast time of our algorithm is not bounded by cycle $C_m$. This will become useful when we assume that some cycle $C_i$, with $i \neq m$, is fully informed in exactly $t_A(u)$ time units. Again, we will consider two cases: $t_A(u) \geq 2k+d-1$ and $t_A(u) < 2k+d-1$. 

\paragraph*{Case 1 ($t_A(u) \geq 2k+d-1$):}

By the lower bound of Lemma~\ref{lma:d_gt_0_lowerbound}, we have $t_{opt}(u) \geq d + \ceil*{\frac{l_i+2i-2}{2}} \geq d + \frac{l_i+2i-2}{2}$ where $1 \leq i \leq k$, $i \neq m$. This gives us
\begin{equation}\label{eq:l_i_lowerbound_case_2_new}
    2t_{opt}(u)-2d-2i+2 \geq l_i
\end{equation}
In our approximation algorithm, all cycles $C_i$, $i \neq m$, $1 \leq i \leq k$ are informed by $v_c$ by time units $d+i$ and $d+i+k-1$, irrespective of their size relative to cycle $C_m$. By the second call from $v_c$, $C_i$ has $k-1$ informed vertices. Starting at time unit $d+i+k-1$, $C_i$ has two extra vertices informed at every time unit. So, all cycles $C_i$ are fully informed at time unit $t_i = d + k + i - 2 + \ceil*{\frac{l_i-(k-1)}{2}} = \ceil*{\frac{2i+k+2d-3+l_i}{2}}$. Thus, substituting (\ref{eq:l_i_lowerbound_case_2_new}), we get
\begin{align*}
    &t_i \leq \ceil*{\frac{2i+k+2d-3+2t_{opt}(u)-2d-2i+2}{2}} = \ceil*{\frac{k+2t_{opt}(u)-1}{2}} \\
    &t_i \leq \ceil*{\frac{k-1}{2}} + t_{opt}(u)
\end{align*}
Since $t_A(u) = \max_{1 \leq i \leq k}(t_i)$, we have that $t_A(u) \leq \ceil*{\frac{k-1}{2}}+t_{opt}(u)$. Also since $t_{opt}(u)\geq t_{opt}(v_c)$, from the first lower bound of Lemma~\ref{lma:lower_bounds_lemma}, we have $t_{opt}(u)-1 \geq k$. So,

\begin{align*}
    &\frac{t_A(u)}{t_{opt}(u)} \leq \frac{t_{opt}(u) + \ceil*{\frac{k-1}{2}}}{t_{opt}(u)} \leq 1 + \frac{\ceil*{\frac{t_{opt}(u)-2}{2}}}{t_{opt}(u)} \leq 1 + \frac{\frac{t_{opt}(u)-1}{2}}{t_{opt}(u)} \\ &\leq 1 + \frac{0.5t_{opt}(u)-0.5}{t_{opt}(u)} = 1.5 - \frac{0.5}{t_{opt}(u)} < 1.5
\end{align*}

So, \textsc{Simple-K-Cycle} is a $(1.5-\epsilon)$-approximation algorithm when $t_A \geq 2k+d-1$.

\paragraph*{Case 2 ($t_A(u) < 2k+d-1$):}

We will prove this case by contradiction. If $t_A(u) = k+p+d-1$, for some $1 \leq p < k$, then cycles $C_1, ..., C_p$ are informed by $v_c$ twice, while cycles $C_{p+1}, ... C_k$ are only informed once. Cycles $C_1, ..., C_p$ are informed exactly like in case 1, so they must finish before $1.5t_{opt}(u)$. So, consider cases where the broadcast time is bounded by some cycle $C_{p+1}, ..., C_k$.

    Let $i$ be the smallest value $p+1 \leq i \leq k$ such that cycle $C_i$ is only fully informed at time unit $t_A(u)=p+k+d-1$. For the sake of contradiction, assume $t_A(u)=k+p+d-1 \geq 1.5t_{opt}(u)$. In this case, cycle $C_i$ is fully informed at exactly $t_A(u)$, so $t_i=t_A(u)$. Note that all cycles $C_j$, $p+1 \leq j \leq k$ receive a call from $v_c$ by time unit $d+j$ independent of whether $j < m$ or $j > m$. So, we can assume that cycle $C_i$ receives a call from $v_c$ at time $d+i$. Note that proving this case implicitly proves the case where $C_i$ is informed earlier than time unit $d+i$; this is the case when $m<i$, causing $C_i$ to be informed in time unit $d+i-1$. So, $t_A (u)= d+i-1+l_i$. Using (\ref{eq:l_i_lowerbound_case_2_new}) and the above, we get
    \begin{align*}
        &t_A(u)=p+k+d-1=d+i-1+l_i\\
        &t_A(u)=p+k+d-1\leq d+i-1+2t_{opt}(u)-2i-2d+2\\
        &t_A(u)=p+k+d-1\leq 2t_{opt}(u)-i-d+1
    \end{align*}
    From our assumption that $k+p+d-1 \geq 1.5t_{opt}(u)$ and above,
    \begin{align*}
        &1.5t_{opt}(u) \leq p+k+d-1 \leq 2t_{opt}(u)-i-d+1
    \end{align*}
    Then, $1.5t_{opt}(u) \leq 2t_{opt}(u)-i-d+1$ which leads to $i \leq \frac{t_{opt}(u)}{2}-d+1$.
    However, from our selection of $i$ we know that $p+1 \leq i$. Transitively, this means that $p+1 \leq \frac{t_{opt}(u)}{2}-d+1$. Thus, $p \leq \frac{t_{opt}(u)}{2}-d$. Also, since $t_{opt}(u) \geq k+1$, we have that $k \leq t_{opt}(u)-1$. Thus, substituting $p$ and $k$ from above,
    \begin{align*}
        &t_A(u) = p+k+d-1 \leq \frac{t_{opt}(u)}{2} - d+t_{opt}(u)-1+d-1 = 1.5t_{opt}(u)-2
    \end{align*}
    This is a contradiction with the assumption that $t_A(u) \geq 1.5t_{opt}(u)$, so $\frac{t_A(u)}{t_{opt}(u)} < 1.5$ when $t_A(u) < 2k+d-1$. Therefore, \textsc{Simple-K-Cycle} is a $(1.5-\epsilon)$-approximation algorithm when $u \neq v_c$.
\end{proof}
\section{Performance of \textsc{Simple-K-Cycle} in Specific Cases}\label{sec:particular-cases}
In this section, we present two examples of graphs that demonstrate that our approximation ratio calculated for \textsc{Simple-k-cycle} is tight, and that \textsc{Simple-k-cycle} is optimal in some cases. For both examples, let $t_A(v_c)$ be the broadcast time of $v_c$ in $G$ under the broadcast scheme generated by \textsc{Simple-K-Cycle}. Also, let $t_{opt}(v_c)$ be the optimal broadcast time of $v_c$ in $G$. In both examples, we use the fact that any broadcast scheme in which all cycles are fully informed in exactly the last time unit is optimal, which we state without a full proof. 

\begin{theorem}
    The approximation ratio of the \textsc{Simple-K-Cycle} algorithm originating at $v_c$ asymptotically approaches $1.5$ for $k$-cycle graphs having $l_1=2k+2$ and $l_2=l_3=...=l_k=2$ as the value of $k$ increases.
\end{theorem}

\begin{proof}
If $l_1=2k+2$ and $l_2=l_3=...=l_k=2$, then the optimal broadcast scheme is to inform $C_1$ in time units 1 and 2, then inform each cycle $C_i$ in time unit $i+1$. This leads to cycles $C_2, C_3, ..., C_{k-1}$ being fully informed before time unit $k+2$, and cycles $C_1, C_k$ being fully informed in exactly time unit $k+2$. It is simple to see why $C_k$ is fully informed in this time unit, for cycle $C_1$, the time unit in which it is fully informed is calculated as $t_{opt}(v_c)=1+\ceil{\frac{l_1-1}{2}}=1+\ceil{\frac{2k+2-1}{2}}=\ceil{\frac{2k+3}{2}}=k+2$.

In our approximation algorithm, all cycles $C_2, C_3, ..., C_k$ will be fully informed by time unit $k+1$. Cycle $C_1$ will be fully informed in time unit $t_{A}(v_c)=k+\ceil{\frac{l_1-k}{2}}=k+\ceil{\frac{2k+2-k}{2}}=\ceil{\frac{3k+2}{2}}$. It follows that $\frac{3k+2}{2} \leq t_{A}(v_c)\leq \frac{3k+3}{2}$. Then, we get the following:

\begin{align*}
    &\frac{\frac{3k+2}{2}}{k+2}\leq \frac{t_A(v_c)}{t_{opt}(v_c)} \leq \frac{\frac{3k+3}{2}}{k+2}\\
    &\frac{3k+2}{2k+4}\leq \frac{t_A(v_c)}{t_{opt}(v_c)} \leq \frac{3k+3}{2k+4}
\end{align*}

$\frac{3k+3}{2k+4} < 1.5$ for all $k\geq 1$, and $\lim_{k \to \infty}\frac{3k+2}{2k+4}=\lim_{k \to \infty}\frac{3k+3}{2k+4}=\frac{3}{2}=1.5$. By the squeeze theorem, the ratio $\frac{t_A(v_c)}{t_{opt}(v_c)}$ also approaches $1.5$ as $k$ approaches infinity. So, for this particular graph family, our algorithm's approximation ratio approaches, but stays strictly below $1.5$. 

\end{proof}

\begin{theorem}
    For any integer j, where $1 \leq j < k$, \textsc{Simple-k-cycle}, with originator $v_c$, is an exact algorithm for any $k$-cycle  graph $G$ whose cycle lengths are structured as follows:
        \item $l_i-l_{i+1}=2$ for all $1 \leq i \leq j$
        \item $l_i-l_{i+1}=1$ for all $j < i < k$
\end{theorem}

\begin{proof}

Let $t_A = p+k$, where $1 \leq p$. In this case, all cycles $C_i$, $1 \leq i \leq p$, $i \leq k$ are informed twice, while all other cycles ($C_j$, $p < j \leq k$), are informed only once.
All cycles $C_i$ are informed in time units $i$ and $k+i$. Since one vertex on cycle $C_i$ is informed in each time unit between $i$ and $k+i$ and two vertices are informed in each time unit after time unit $k+i$, we can determine the total number of informed vertices by time unit $t_A$ by summing the vertices informed before time unit $k+i$ ($k$ vertices) and after time unit $k+i$: $2((p+k)-(k+i-1))=2(p-i+1)$ vertices. So, $C_i$ has $k+2(p-i+1)=k+2p-2i+2$ informed vertices by time unit $k+p=t_A$. By the same logic, every cycle $C_j$ is informed only in time unit $j$. So, they have $p+k-(j-1)=p+k-j+1$ informed vertices by time unit $p+k$. 

As a result, if all $l_i=k+2p-2i+2$ and all $l_j=p+k-j+1$ for some fixed $p$, all cycles will be informed in exactly $k+p$ time units, and the broadcast scheme will be optimal. We can then see that a set of cycle lengths follows this pattern iff 
    \item $l_i-l_{i+1}=(k+2p-2i+2)-(k+2p-2(i+1)+2)=2$ for all $1 \leq i < p, i<k$.
    \item $l_j-l_{j+1}=(p+k-j+1)-(p+k-(j+1)+1)=1$ for all $p< j < k$.
    \item $l_p-l_{p+1}=(k+2p-2p+2)-(p+k-(p+1)+1)=2$ if $p < k$.
\end{proof}
\section{Conclusion}
\label{sec:conclusion}

In this paper, we designed a simple, linear-time algorithm for the broadcast time problem in general $k$-cycle graphs. Further, we showed that the given algorithm is a $(1.5-\epsilon)$-approximation algorithm. This improves the previously best known fixed-ratio approximation guarantee of 2, found in \cite{k-cycle-2-approximation} and generalized for cactus graphs in \cite{cactus-telephone-broadcasting}. It also represents the simplest algorithm with the lowest time-complexity for achieving a 1.5-approximation, even considering the recent PTAS found in \cite{bringolf2026hardnessapproximationbroadcastingstructured}. Moreover, due to the known NP-completeness of the problem \cite{bringolf2026hardnessapproximationbroadcastingstructured} and the existence of a PTAS, a natural direction for future work is to improve the approximation ratio of efficient fixed-ratio algorithms. Another future direction is to extend this $(1.5-\epsilon)$-approximation ratio to general cactus graphs.
\begin{credits}

\subsubsection{\discintname}
The authors have no competing interests to declare that are relevant to the content of this article.
\end{credits}

\bibliography{bibliography}

\end{document}